\documentclass[11pt,letterpaper]{article}
\usepackage{color,ifpdf,latexsym,graphicx,url}
\usepackage[margin=1in]{geometry}
\usepackage{amsmath,amssymb,amsthm}
\usepackage{subcaption}

\title{Orthogonal Fold \& Cut}
\author{%
  \tabcolsep=1.5em
  \begin{tabular}{cccc}
    Hayashi Ani%
      \thanks{MIT Computer Science and Artificial Intelligence Laboratory,
        32 Vassar Street, Cambridge, MA 02139, USA,
        \protect\url{{brunnerj,edemaine,mdemaine,dylanhen,vluo,rachanam}@mit.edu}}
    &
    Josh Brunner\footnotemark[1]
    &
    Erik D. Demaine\footnotemark[1]
    &
    Martin L. Demaine\footnotemark[1]
  \end{tabular}\and
  \tabcolsep=1.5em
  \begin{tabular}{ccc}
    Dylan Hendrickson\footnotemark[1]
    &
    Victor Luo\footnotemark[1]
    &
    Rachana Madhukara\footnotemark[1]
  \end{tabular}
}
\date{}

\newif\ifabstract
\abstracttrue
\newif\iffull
\ifabstract \fullfalse \else \fulltrue \fi

\usepackage{hyperref}
\hypersetup{breaklinks,bookmarksnumbered,bookmarksopen,bookmarksopenlevel=2}
{\makeatletter \hypersetup{pdftitle={\@title}}}

{\makeatletter
 \gdef\xxxmark{%
   \expandafter\ifx\csname @mpargs\endcsname\relax 
     \expandafter\ifx\csname @captype\endcsname\relax 
       \marginpar{xxx}
     \else
       xxx 
     \fi
   \else
     xxx 
   \fi}
 \gdef\xxx{\@ifnextchar[\xxx@lab\xxx@nolab}
 \long\gdef\xxx@lab[#1]#2{\textbf{[\xxxmark #2 ---{\sc #1}]}}
 \long\gdef\xxx@nolab#1{\textbf{[\xxxmark #1]}}
}

{\makeatletter \gdef\fps@figure{!htbp}}

\let\realbfseries=\bfseries
\def\bfseries{\realbfseries\boldmath}

\newtheorem{theorem}{Theorem}[section]
\newtheorem{lemma}[theorem]{Lemma}
\newtheorem{proposition}[theorem]{Proposition}
\newtheorem{corollary}[theorem]{Corollary}

\let\epsilon=\varepsilon
\def\defn#1{\textbf{\textit{\boldmath #1}}}

\begin{document}
\maketitle

\begin{abstract}
  We characterize the cut patterns that can be produced by
  ``orthogonal fold \& cut'':
  folding an axis-aligned rectangular sheet of paper
  along horizontal and vertical creases,
  and then making a single straight cut (at any angle).
  Along the way, we solve a handful of related problems:
  orthogonal fold \& punch, 1D fold \& cut, signed 1D fold \& cut,
  and 1D interval fold \& cut.
\end{abstract}

\section{Introduction}
\label{sec:intro}

Given a rectangular piece of paper marked with a pattern of line segments called
\defn{cuts}, the \defn{fold \& cut} problem asks us to fold the paper into a
flat origami and then make a single infinite cut to cut exactly along the
desired pattern of line segments.
More formally, \defn{making a single infinite cut} is the operation of
subtracting a given \defn{cut line}, i.e.,
removing from the folded paper all points that lie along the cut line;
in particular, the cut line may overlap folded creases and
may meet the folded paper at one or more line segments,
some of which may have zero length (i.e., be individual points).%
\footnote{We normally view the piece of paper as an open set
  (excluding its boundary), so that the connected components resulting from a
  cut are also open sets, but this is not essential.}
%
The goal of the fold \& cut problem is to fold the given piece of paper
such that the specified cuts all lie on the same line,
and no other points of the paper lie on that line, so that
making a single infinite cut along that cut line
produces the desired set of cuts (some of which may have zero length).
It is known that fold \& cut is solvable for any (finite) set of line segments
\cite{bern2002disk,demaine1998folding,Demaine-O'Rourke-2007}.

The \defn{orthogonal fold \& cut} problem asks the same question,
but restricted to \defn{orthogonal} folds,
meaning that every crease must be parallel to an edge
of the rectangle; we orient the paper to be \defn{axis-aligned}
so that the edges and creases are horizontal and vertical in~$\mathbb R^2$.
Again we allow only a single infinite cut, which may be at any angle.
It is easy to construct sets of cuts that cannot be obtained this way,
but some complicated-looking shapes still can be. 
Figure~\ref{fig:example} shows an example instance and its solution;
Figures~\ref{fig:puzzle font} and~\ref{fig:font} show several more examples
(but also some examples where the paper is not an axis-aligned rectangle,
so not an instance of orthogonal fold \& cut).

\begin{figure}
  \centering
  \includegraphics{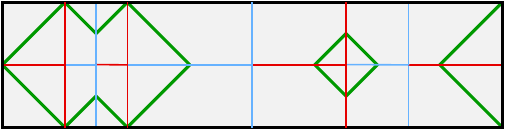}
  \caption{An example instance of orthogonal fold \& cut (bold green lines), and the crease pattern our algorithm generates (thin red mountains and blue valleys).}
  \label{fig:example}
\end{figure}

\begin{figure}
  \centering
  \includegraphics[width=\linewidth]{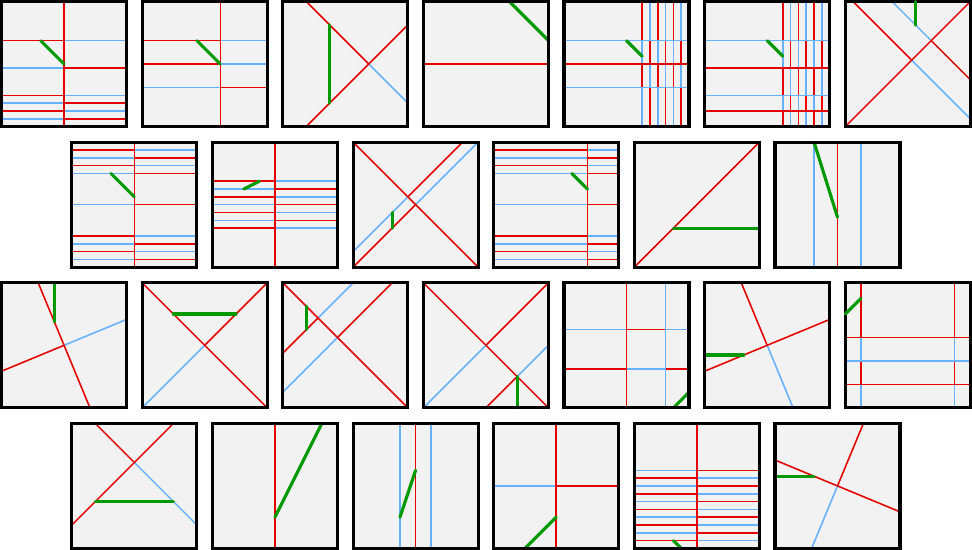}
  \caption{Our orthogonal fold \& cut puzzle font.
    If you fold along the thin lines and cut along the green line,
    what shapes do you get? (Solutions in Figure~\ref{fig:font}.)}
  \label{fig:puzzle font}
\end{figure}

In this paper, we characterize when the orthogonal fold \& cut problem
is solvable, that is, which sets of cuts on an axis-aligned rectangle of paper
can be obtained by horizontal and verticals folds and a single infinite cut.
In addition, we provide characterizations for several simpler related problems:
\begin{enumerate}
\item Given a line segment of paper marked with a pattern of \defn{cut points}, \defn{1D fold \& cut} asks us to flat-fold the paper segment and make a single (0-dimensional) point cut that hits exactly the cut points.
  Even if we do not allow folds at cut points,
  this problem is always solvable; see Section~\ref{sec:1d basic}.
\item In \defn{signed 1D fold \& cut}, each cut point is given a sign, either positive or negative, and we are asked to have the folded paper oriented according to the sign at each cut point.
Intuitively, the cut points are marked on only one side, and they need to all be face up.
  This problem is solvable if and only if the signs alternate;
  see Section~\ref{sec:1d signed}.
\item In \defn{1D interval fold \& cut}, we are asked to flat-fold a line segment of paper so that specified \defn{cut intervals} lie on a common interval in the folded state, and no other part of the paper segment folds to this interval. Additionally, the cut intervals contain some marked creases that are required to be folded, and no other creases can be made within any cut interval.
  We give an algorithmic characterization of solvability in
  Section~\ref{sec:1d interval}.
\item Returning to 2D paper, in \defn{orthogonal fold \& punch},
  we are given a rectangular piece of paper marked with points
  called \defn{holes}, and are asked to fold the paper flat
  using only orthogonal folds and then punch out a single point
  to remove exactly the specified holes.
  It is known that the nonorthogonal \defn{fold \& punch} problem
  always has a solution, in particular because it is a special case of
  fold \& cut \cite{asao2017folding}.
  For the orthogonal version, we give an algorithmic characterization of
  solvability in Section~\ref{sec:of&p} as a warmup to our solution to
  orthogonal fold \& cut.
\end{enumerate}
Although we allow creases to be along cuts in 2D paper,
or at cut points in 1D paper, none of our solutions do so.
Thus, all of our results also hold in the \defn{scissor-cut} model
\cite{demaine1998folding,Demaine-O'Rourke-2007}: in the flat folding,
every cut has material on both sides (as needed for cutting with scissors).
This is in contrast to general fold \& cut, where some instances have
solutions but not with scissor cuts; in the problems we study,
the models are equivalent.

\begin{figure}
  \centering
  \includegraphics[width=\linewidth]{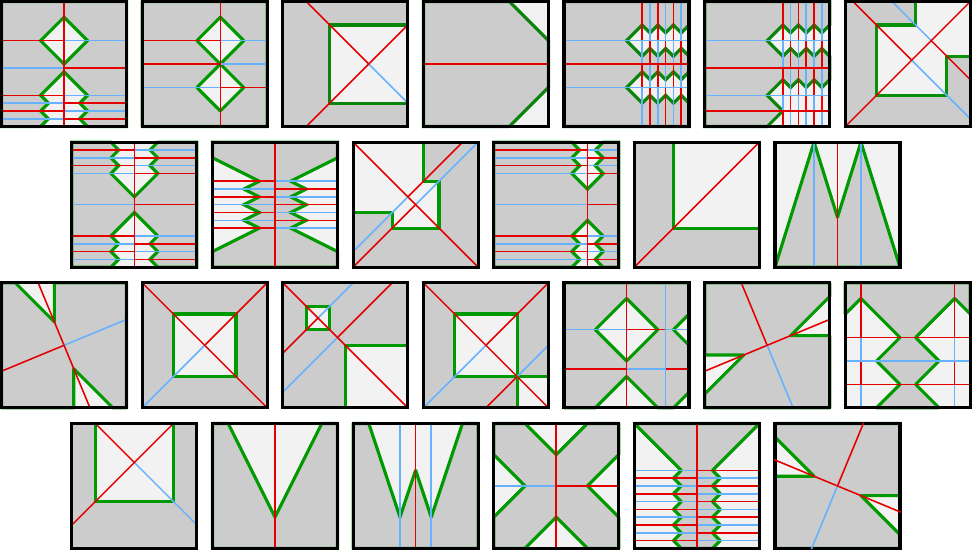}
  \caption{Our orthogonal fold \& cut mathematical font / solutions to the puzzles in Figure~\ref{fig:puzzle font}, in the notation of Figure~\ref{fig:example}.
    Note that C, G, J, L, N, O, P, Q, S, U, and Z have folds that are not
    parallel to the edges of the paper, so only
    A, B, D, E, F, H, I, K, M, R, T, V, W, X, and Y
    are proper instances of the orthogonal fold \& cut problem
    solved in this paper.}
  \label{fig:font}
\end{figure}

Our characterizations take the form of defining an easy-to-compute \defn{canonical} crease pattern and showing that, if the problem has any solution, then the canonical solution works.
For example, the canonical crease pattern for unsigned and signed 1D fold \& cut puts a crease point at the midpoint between each consecutive pair of cut points.
For 1D interval fold \& cut, the canonical crease pattern has a crease at each required crease and at the midpoint between each pair of consecutive cut intervals.

For the main problem, orthogonal fold \& cut, we first note that the slopes of
cut lines in a solvable instance must all be $\pm\alpha$ for a common $\alpha$.
In the degenerate cases $\alpha=0$ and $\infty$, all cuts are parallel,
and we show that the problem reduces to 1D fold \& cut.
Otherwise, we can scale the paper so $\alpha=1$ (i.e., all cuts have
an angle of $\pm 45^\circ$ to the axes), as in Figure~\ref{fig:example}.
Our main result characterizes the solvable instances of orthogonal fold
\& cut in terms of ``stripes'' and ``bands'', by factoring the instance into two
orthogonal instances of 1D interval fold \& cut.
See Section~\ref{sec:of&c} for details,
including the definition of the canonical crease pattern.

We also designed a mathematical font based on orthogonal fold \& cut,
where the creases all meet orthogonally to each other and
the cut-out shapes form letters of the alphabet, as shown in
Figure~\ref{fig:font}.%
\footnote{See also \url{https://erikdemaine.org/fonts/orthofoldcut/}
  for a web app where you can write your own messages in the font.}
The crease patterns are canonical, as output by our algorithm.
The cut patterns are all designed by hand to guarantee solvability.
By hiding all but one of the cuts, as in
Figure~\ref{fig:puzzle font}, we also obtain a \defn{puzzle font},
where reading each letter requires imagining the result of the fold \& cut
process.  Notably, while our font uses a consistent square shape of paper,
the paper is not an axis-aligned rectangle in the sense of having sides
parallel to the allowed crease lines.  We discuss this open problem further in
Section~\ref{sec:open}.

\section{1D Fold \& Cut Problems}\label{sec:1d}

\subsection{1D Fold \& Cut}
\label{sec:1d basic}

We first consider the simpler problem of fold \& cut with one-dimensional paper,
called \defn{1D fold \& cut}.
Given a line segment of paper with marked \defn{cut points},
this problem asks us to fold the paper flat (meaning, onto a line)
to align the given cut points onto a common point,
with no other folded paper overlapping that point.
Thus, if we make a single (0-dimensional) cut
by removing all points of paper that overlap that common point,
we end up removing exactly the specified cut points.
We show this problem always has a solution,
even if we forbid folds at cut points.

To solve 1D fold \& cut, we claim that it suffices to specify the locations of creases.
Specifying the locations of unsigned creases uniquely determines the position of each point in the segment in the folded state.
Furthermore, any unsigned 1D crease pattern can be folded flat, e.g., with an \defn{accordion} fold that alternates mountain and valley \cite[Section~12.1]{Demaine-O'Rourke-2007}.
We use this approach to prove universality:

\begin{lemma}\label{lem:1df&c}
  Every instance of 1D fold \& cut is solvable, without creasing at cut points.
\end{lemma}

\begin{proof}
  Refer to Figure~\ref{fig:1d}.
  We put a crease at the midpoint between each pair of consecutive cut points.
  This crease aligns consecutive cut points, so together they align all cut points.
  In the folded state, each segment between consecutive cut points or between an end of the paper and the nearest cut point lies entirely on one side of the aligned cut points; thus exactly the cut points are aligned.
\end{proof}

\begin{figure}
  \centering
  \begin{subfigure}{\linewidth}
    \centering
    \includegraphics{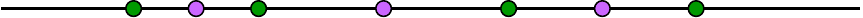}
  \end{subfigure}
  $$\downarrow$$
  \begin{subfigure}{\linewidth}
    \centering
    \includegraphics{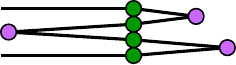}
  \end{subfigure}
  \caption{An instance of 1D fold \& cut and its canonical solution. The green points are cut points, and the purple points are creases in the canonical solution.}
  \label{fig:1d}
\end{figure}

Creasing the midpoint between consecutive cut points is in some sense the only option we have. There are solutions that have additional creases between cut points, but these extra creases do not accomplish anything: we could pull the paper out, leaving only a crease at the midpoint, while only locally (to the interval between cut points) changing the folded state and not affecting the aligned cut points.
Similarly, there are solutions that add creases between an extreme (first or
last) cut point and an adjacent paper endpoint, but these creases can be
removed while only locally changing the folded state
and not affecting the aligned cut points.
We call the solution that folds at midpoints the \defn{canonical} solution. Figure~\ref{fig:1d} shows an example instance and its canonical solution.

\subsection{Signed 1D Fold \& Cut}
\label{sec:1d signed}

\defn{Signed 1D fold \& cut} is an extension of 1D fold \& cut
where each cut point is given a \defn{sign}, either positive or negative.
We require any solution to have the additional property that the folded paper
is oriented in the direction according to the sign at each cut point.
One can imagine that each cut point is marked on only one side of the paper,
and we want the marks to all be face up when we cut.
This problem is a good model for orthogonal fold \& cut on a narrow strip, where the sign of a cut point is the same as the sign of the slope of a cut segment crossing the strip.
In this problem, we must explicitly forbid creases at cut points, because
the orientation of the paper is not well-defined at a crease.
Because the paper orientation changes at crease points,
independent of whether they are mountain or valley,
it again suffices to specify the locations of creases for this problem.

\begin{lemma}\label{lem:s1df&c}
  An instance of signed 1D fold \& cut is solvable if and only if
  the signs of cuts alternate along the segment of paper.
\end{lemma}

\begin{proof}
  Suppose there are two consecutive cut points with the same sign,
  which without loss of generality is positive.
  Let $a < b$ be the two positions of these cut points,
  viewing the segment of paper as a subset of the real line~$\mathbb R$,
  with the positive orientation being in the positive direction.
  Let $f(x)$ be the location of point $x$ in the folded state,
  which we also view as lying on the real line $\mathbb R$.
  Aligning the cut points requires $f(a)=f(b)$.
  That both cut points have positive sign requires that
  $f$ is increasing at both $a$ and~$b$.
  In particular, for some $\varepsilon>0$, we have
  $f(b-\varepsilon)<f(b)=f(a)<f(a+\varepsilon)$ and
  $a+\varepsilon < b-\varepsilon$.
  The Intermediate Value Theorem implies that there is a point $c$
  with $a < c < b$ and $f(c)=f(a)$.
  Thus this folding aligns a third point $c$ with $a$ and~$b$,
  which cannot be a crease point because $a$ and $b$ were assumed to be consecutive,
  and thus the folding is not a valid solution for aligning exactly the
  cut points.
  Therefore cut points having alternating signs is necessary for signed 1D fold \& cut to be solvable.

  Now suppose the signs do alternate.
  Then the canonical solution from Lemma~\ref{lem:1df&c} is still a solution,
  assuming we rotate the folded state as needed to orient the initial segment of paper to have the
  correct sign: there is one crease between each pair of consecutive cut
  points, so the orientation of the paper is opposite at the two cut points.
\end{proof}

\subsection{1D Interval Fold \& Cut}
\label{sec:1d interval}

In \defn{1D interval fold \& cut}, we are given a line segment of paper
together with some open intervals called \defn{cut intervals},
and we are asked to find a flat folding that places all cut intervals onto
a common open interval called the \defn{final cut interval},
with no other part of the paper segment folding to the final cut interval,
except possibly for creased endpoints shared by two cut intervals.
Define \defn{empty intervals} to be the intervals of paper resulting
from removing the cut intervals,
which are closed except at the ends of the paper;
thus the flat folding must place every empty interval of positive length
outside of the final cut interval,
and have a crease at every zero-length empty interval.%
%
\footnote{In Section~\ref{sec:intro}, we viewed zero-length empty intervals
  as required creases within cut intervals, which is equivalent.}
Additionally, no creases are permitted within any cut interval.
Figure~\ref{fig:1dinterval} shows an example instance and solution.

\begin{figure}
  \centering
  \begin{subfigure}{\linewidth}
    \centering
    \includegraphics[scale=1]{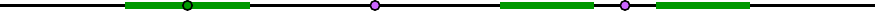}
  \end{subfigure}
  $$\downarrow$$
  \begin{subfigure}{\linewidth}
    \centering
    \includegraphics[scale=1]{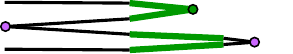}
  \end{subfigure}
    \caption{An example instance and solution of 1D interval fold \& cut.
      Empty intervals are black; the cut intervals are green;
      the green dot is a required fold in a zero-length empty interval;
      and the purple dots are the other added folds in the solution.}
  \label{fig:1dinterval}
\end{figure}

This problem will end up being a nice analog to our algorithm for orthogonal fold \& cut. We present it as a simpler case for intuition.

Define the \defn{canonical} crease pattern to consist of
a crease at the midpoint of each empty interval
(including zero-length empty intervals)
that is surrounded by two cut intervals.
We prove the following characterization:

\begin{lemma}\label{lem:1d interval}
  An instance of 1D interval fold \& cut is solvable if and only if the canonical crease pattern is a solution.
\end{lemma}

\begin{proof}
Consider an empty interval surrounded by two cut intervals.
If the empty interval has zero length, then it must be creased,
as in the canonical crease pattern.
If the empty interval has positive length, then both its endpoints
must fold to a common endpoint of the final cut interval,
because the cut intervals on each side must lie on the final cut interval,
but the empty interval must not overlap the final cut interval.
By the same argument as 1D fold \& cut in Section~\ref{sec:1d basic},
any solution achieving this can be modified to have only a single crease at the midpoint
of this empty interval without changing anything about the folding outside of
the empty interval in question.
Similarly, an empty interval incident to an endpoint of the paper segment
can be modified to not have any creases:
the endpoint incident to a cut interval (if any) must be folded to an
endpoint of the final cut interval, so the empty interval will safely
continue away from the final cut interval without any creases.
Because creases within cut intervals are not permitted,
we have arrived at the canonical crease pattern.
\end{proof}

Note that, unlike 1D fold \& cut, this problem is not always solvable.
For example, if there are two cut intervals of different lengths surrounded by positive-length empty intervals,
then it is impossible to fold them to perfectly align.

\section{Orthogonal Fold \& Punch}\label{sec:of&p}

We now move on to 2-dimensional paper, and consider orthogonal fold \& punch as a warm-up for orthogonal fold \& cut.

In the \defn{fold \& punch} problem, we are given a rectangular piece of paper
marked with points called \defn{holes}, and are asked to fold it flat to align
the given holes onto a common point called the \defn{punch point},
with no other folded paper overlapping the punch point.
Thus, if we \defn{punch} out the punch point, i.e., remove the points of
folded paper that overlap the punch point, we end up removing exactly the
specified set of holes.
Like fold \& cut, fold \& punch is always solvable; indeed, fold \& punch is
the special case of fold \& cut where all segments have length zero,
but fold \& punch admits simpler solutions \cite{asao2017folding}.

In \defn{orthogonal fold \& punch}, we are restricted to orthogonal creases
(parallel to the sides of the rectangle); like orthogonal fold \& cut,
we orient the paper to be an axis-aligned rectangle in $\mathbb R^2$
so that the creases are horizontal and vertical.
In this section, we characterize the patterns of holes
that can be punched out in this setting.
First, in Section~\ref{sec:crease patterns}, we consider general properties of
orthogonal crease patterns; this will also be useful for orthogonal fold \& cut.
Then, in Section~\ref{sec:solving of&p}, we solve orthogonal fold \& punch.

\subsection{Orthogonal Crease Patterns}\label{sec:crease patterns}

We first consider what possible orthogonal crease patterns look like.
As in the 1D case, we show that an orthogonal folding can be described
just by the positions of crease lines.

\begin{lemma}\label{lem:orthogonal crease patterns}
  An unsigned orthogonal crease pattern on rectangular paper that folds flat must consist of horizontal and vertical lines that go all the way across the paper. Moreover, every such set of horizontal and vertical lines can be folded flat.
\end{lemma}

\begin{proof}
  For a crease to end inside the paper, it must end in a vertex that folds flat. But the only nontrivial vertex that satisfies the Kawasaki criterion with only horizontal and vertical creases is the degree-4 vertex where a horizontal and a vertical line cross. This vertex cannot end a horizontal or vertical crease, so they must continue to the edge of the paper.

  To fold flat a crease pattern consisting of horizontal and vertical lines, we essentially solve the 1D problem twice. First, accordion fold the vertical creases. Each horizontal crease still lies on a horizontal line, so we can now accordion fold the horizontal creases.
\end{proof}

An unsigned crease pattern uniquely determines the location of points in the folded state, so solving orthogonal fold \& punch or orthogonal fold \& cut amounts to specifying the position of some horizontal lines and some vertical lines; we do not need to consider the orientation of each crease.

Next, we prove a crucial fact about the folded state of an orthogonal crease pattern, which will be relevant to our characterizations for both orthogonal fold \& punch and orthogonal fold \& cut.

\begin{lemma}\label{lem:orthogonal align}
  Suppose two points $(x_1,y_1)$ and $(x_2,y_2)$
  on an axis-aligned rectangle of paper
  are aligned in the folded state of a crease pattern
  with only horizontal and vertical creases.
  Then $(x_1,y_2)$ and $(x_2,y_1)$ are also aligned with them.
\end{lemma}

\begin{proof}
  Because creases are all either horizontal or vertical, each vertical line on the paper lies on a vertical line in the folded state. In other words, the horizontal position of a point in the folded state depends only on its horizontal position on the paper. Let $h(x)$ be the horizontal position in the folded state of a point at horizontal position $x$ on the paper. Similarly, let $v(y)$ be the vertical position in the folded state of a point at vertical position $y$ on the paper, so $(x,y)$ folds to $(h(x),v(y))$. That the folded state factors in this way is the crucial property of orthogonal foldings of rectangular (or more generally, orthogonally convex) paper.

  Because $(x_1,y_1)$ and $(x_2,y_2)$ are aligned at $(h(x_1),v(y_1))$, we must have $h(x_1)=h(x_2)$ and $v(y_1)=v(y_2)$. Thus $(x_1,y_2)$ and $(x_2,y_1)$ also lie on $(h(x_1),v(y_1))$ in the folded state.
\end{proof}

\subsection{Solving Orthogonal Fold \& Punch}\label{sec:solving of&p}

With Lemma~\ref{lem:orthogonal align} at our disposal, it is straightforward to characterize the solvable hole patterns for orthogonal fold \& punch.

\begin{theorem}
  An instance of orthogonal fold \& punch is solvable if and only if the set of holes is a combinatorial rectangle. That is, the set of holes must be of the form $X\times Y$, where $X$ and $Y$ are finite sets of horizontal and vertical coordinates, respectively.
\end{theorem}

\begin{proof}
  Suppose the set of holes is $X\times Y$. We solve the instance by solving $X$ and $Y$ separately as instances of 1D fold \& cut. That is, we put a vertical crease at the midpoint between consecutive elements of $X$, and a horizontal crease at the midpoint between consecutive elements of $Y$. These creases align exactly the vertical lines at positions in $X$ and the horizontal lines at positions in $Y$, so they align exactly the points in $X\times Y$.

  Now suppose an instance of orthogonal fold \& punch is solvable, and let $X$ and $Y$ be the sets of horizontal and vertical coordinates of points, respectively. 
  Let $x\in X$ and $y\in Y$. Let $(x,y^\prime)$ and $(x^\prime,y)$ be holes; these exist by the definitions of $X$ and $Y$. The folded state aligns $(x,y^\prime)$ and $(x^\prime,y)$, so by Lemma~\ref{lem:orthogonal align} it also aligns $(x,y)$ with them. Thus $(x,y)$ must also be a hole, so the set of holes is all of $X\times Y$.
\end{proof}

\section{Orthogonal Fold \& Cut}\label{sec:of&c}

In this section, we characterize the patterns of cuts that can be obtained in orthogonal fold \& cut. In Section~\ref{sec:slopes}, we investigate necessary properties of a set of cuts by considering the slopes of the cut. Then, in Section~\ref{sec:finite slope} we characterize the solvable instances of orthogonal fold \& cut that satisfy these properties.

\subsection{Slopes of Cuts}\label{sec:slopes}

Considering slopes of cut lines immediately gives a significant constraint on solvable instances of orthogonal fold \& cut.

\begin{lemma}\label{lem:slopes}
  If an instance of orthogonal fold \& cut contains cuts of slopes $\alpha$ and $\beta$ with $\alpha\neq\pm\beta$, it is not solvable.
\end{lemma}

\begin{proof}
  Reflecting across a horizontal or vertical line negates slopes. Thus a cut of slope $\alpha$ must lie on a line of slope $\pm\alpha$ in the folded state. If $\alpha\neq\pm\beta$, these two cuts cannot lie on the same line.
\end{proof}

For the rest of this section, we will assume the slope of every cut line is $\pm\alpha$.
We divide into two cases: $\alpha \in \{0,\infty\}$ and $0 < \alpha < \infty$.

For $\alpha=0$ or $\infty$, the cut pattern consists of entirely horizontal
or of entirely vertical segments respectively.
We show that these cases essentially reduce to 1D fold \& cut:

\begin{lemma}\label{lem:horizontal cuts}
  An instance of orthogonal fold \& cut consisting of only horizontal or only vertical cuts is solvable if and only if every cut goes all the way across the paper.
  Solvable instances do not need to crease along any cuts.
\end{lemma}

\begin{proof}
  Suppose without loss of generality that we have only horizontal cuts which go all the way across the paper.
  Projecting the cuts onto the vertical axis, we obtain an instance of
  1D fold \& cut.
  Applying the solution of Lemma~\ref{lem:1df&c} and extruding horizontally,
  we obtain a solution by accordion folding (alternating mountain and valley)
  with a horizontal crease halfway between each pair of consecutive cuts.

  Conversely, consider a horizontal cut; we will show it must cross the paper.
  The infinite horizontal line $\ell$ containing the cut must lie entirely on some horizontal line $\ell^\prime$ in the folded state, because we use only orthogonal creases.
  Then $\ell^\prime$ contains the cut in the folded state, so we must cut along $\ell^\prime$.
  Doing so cuts every point in the paper in $\ell$, so we must cut the entire intersection of $\ell$ with the paper.
\end{proof}

At this point, we can argue that the \defn{scissor-cut} model,
where every cut locally has material on both of its sides,
does not lose generality in orthogonal fold \& cut.

\begin{corollary} \label{cor:scissor cuts}
  Every solvable instance of orthogonal fold \& cut has a solution where
  no crease has a positive-length overlap with a cut.
\end{corollary}

\begin{proof}
  If $\alpha \notin \{0, \infty\}$, then every crease is axis-aligned
  while every cut is not, so they cannot overlap at more than a single point.
  Otherwise, the claim follows from Lemma~\ref{lem:horizontal cuts}.
\end{proof}

For the remaining case of nonzero finite $\alpha$,
we can scale the paper vertically by
$1/\alpha$ to make the slopes $\pm1$.
This transformation preserves orthogonal foldings
and thus solutions to orthogonal fold \& cut.
So for the rest of this section,
we assume without loss of generality that $\alpha=1$.

\subsection{Non-Axis-Aligned Cut Patterns}\label{sec:finite slope}

In this section, we solve the case where all cuts have slope $\pm1$.

Define a \defn{cut vertex} to be a point interior to the piece of paper
that is an endpoint of a cut or an intersection between noncollinear cuts.
Equivalently, we can subdivide cuts at intersection points between noncollinear
cuts and at endpoints of collinear cuts,
so that the cuts are disjoint except at their endpoints;
then the cut vertices are the endpoints of cuts
except for those on the boundary of the paper.
In this case, we can define the \defn{degree} of a cut vertex to be the
number of incident cuts of positive length sharing this endpoint.
For example, we obtain a degree-$0$ cut vertex
for each isolated cut of zero length,
and a degree-$4$ cut vertex from two properly crossing cuts.
We begin with an observation about cut-vertex degrees,
though we will not need it for our algorithmic characterization.

\begin{proposition}
  In any solvable instance of orthogonal fold \& cut,
  every cut vertex has degree $0$, $2$, or~$4$.
\end{proposition}

\begin{proof}
  Lemma~\ref{lem:slopes} guarantees that every cut vertex
  has maximum degree~$4$, with up to two cuts of slope $\alpha$ and
  up to two cuts of slope $-\alpha$.
  That every vertex has even degree follows from
  Corollary~\ref{cor:scissor cuts}
  (as argued in \cite{demaine1998folding}):
  because no crease lies along a cut, in the folded state every cut is locally surrounded by
  some material on either side of the cut line.
  Because the folded paper cannot meet the cut line except at cuts,
  every open region bounded by cuts and the paper boundary
  must be entirely on one side of the cut line.
  Therefore these regions are 2-colorable,
  with the two colors corresponding to the two sides of the cut line,
  such that two regions adjacent across a cut line have different colors.
  This property implies that all cut vertices have even degree:
  if you walk around a vertex, the region colors must alternate.
  (In fact, face 2-colorability and even vertex degrees are equivalent
  properties \cite{Koenig-1936}.)
\end{proof}

Next we analyze more precisely how creases interact with cuts.

\begin{lemma}\label{lem:no intersection}
  In any solution to an $\alpha = 1$ orthogonal fold \& cut instance,
  a crease cannot intersect a cut other than at a cut vertex.
\end{lemma}

\begin{proof}
  Folding on the crease would align the cut with its reflection across the crease. But if we are away from vertices, and because the cut is not parallel or perpendicular to the crease, the reflection of the cut across the crease does not have a cut. So folding on the crease would align the cut with a noncut, which is not allowed.
\end{proof}

Now we explain how we reduce orthogonal fold \& cut
to various 1D fold \& cut problems.
A thin vertical strip of the paper must lie in a vertical strip in the folded state.
Consider a thin vertical strip containing no vertical creases or cut vertices,
which must fold into a vertical strip of the same width.
In isolation, this strip is essentially an instance of signed 1D fold \& cut, using horizontal creases: a cut line crossing the strip becomes a cut point, and the sign of a cut point depends on the sign of the slope of the cut line. That the cut line takes some vertical distance to cross means we are not allowed to place horizontal creases in a small interval around cut points, but this constraint does not affect solvability: the canonical solution (from Lemma~\ref{lem:s1df&c}) with a crease at the midpoint between each pair of consecutive cuts will not have creases intersecting cuts because the cuts do not intersect in the strip.

For orthogonal fold \& cut to be solvable, each thin vertical or horizontal strip must be solvable as an instance of signed 1D fold \& cut. However, this condition is not sufficient: for instance, each strip in the cut pattern in Figure~\ref{fig:bad} is solvable, but the whole cut pattern is not. To solve orthogonal fold \& cut, we need to solve all of these 1D problems simultaneously.

\begin{figure}
  \centering
  \includegraphics{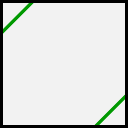}
  \caption{An unsolvable instance of orthogonal fold \& cut. Note that the canonical crease pattern (defined in Theorem~\ref{thm:of&c canonical works}), which would have a horizontal and vertical crease through the center of the square, does not give a valid solution. In order to place the two cut segments on the same line, we would need to also align some portion of a cut segment with part of the paper which is not supposed to be cut; this follows from Lemma~\ref{lem:orthogonal align}.}
  \label{fig:bad}
\end{figure}

We now partition the paper into horizontal intervals based on where
vertical creases are allowed by Lemma~\ref{lem:no intersection}.
For clarity, we consider just vertical creases,
but this discussion applies similarly to horizontal creases.
Each cut (subdivided so it has cut vertices only at its endpoints)
projects to an open horizontal interval in which vertical creases are banned
by Lemma~\ref{lem:no intersection}, so we call then open interval a \defn{band}.
Bands can overlap, in which case we merge them into a single band.
We call each remaining horizontal interval
(which is closed except possibly at the paper boundary)
a \defn{stripe}; creases are allowed within stripes.
Figure~\ref{fig:stripes and bands} shows an example.
Note that a stripe may have \defn{zero width}:
for instance, at a degree-$4$ cut vertex,
creases are allowed through the cut vertex (contained in a zero-width stripe)
but not near it (surrounded by bands).

\begin{figure}
  \centering

  \begin{subfigure}{\linewidth}
    \centering
  \includegraphics{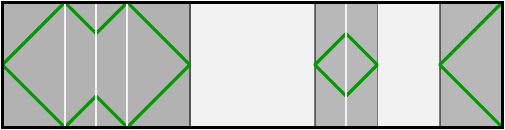}
  \caption{Vertical bands shaded dark, and vertical stripes (including
    zero-width stripes) shaded light.}
  \end{subfigure}

  \medskip

  \begin{subfigure}{\linewidth}
    \centering
    \includegraphics{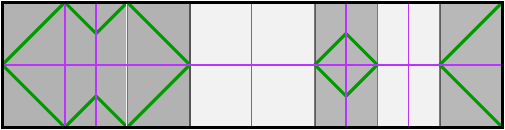}
    \caption{The canonical crease pattern in purple
      (defined in Theorem~\ref{thm:of&c canonical works}).}
  \end{subfigure}
  \caption{Analysis of the instance of orthogonal fold \& cut from Figure~\ref{fig:example}.}
  \label{fig:stripes and bands}
\end{figure}


This partitioning of the paper gives insight into what is needed
from a crease pattern.
In particular, we essentially now have two orthogonal instances of
1D interval fold \& cut, where bands are cut intervals
and stripes are empty intervals.

We prove the following theorem,
analogous to the situation for 1D interval fold \& cut:

\begin{theorem}\label{thm:of&c canonical works}
  Consider an instance of orthogonal fold \& cut
  in which every cut has finite nonzero slope $\pm\alpha$.
  Define the \defn{canonical} crease pattern to
  put one vertical (respectively, horizontal) crease at the center of each
  vertical (respectively, horizontal) stripe, including zero-width stripes.
  If the instance is solvable, then the canonical crease pattern is a solution.
\end{theorem}

\begin{proof}
By Lemma~\ref{lem:orthogonal crease patterns}, it suffices to specify the
locations of the creases, which determines the position of each point
in the folded state.  In particular,
Lemma~\ref{lem:orthogonal crease patterns} gives us an actual folding of
the unassigned canonical crease pattern to test.
It remains to show that any solution can be converted to use the creases
in the canonical crease pattern.

Consider the vertical creases in any solution; horizontal creases are symmetric.
By Lemma~\ref{lem:no intersection}, there cannot be any vertical creases within a vertical band. We consider what the crease pattern can look like within a single vertical stripe, considering zero-width stripes and then positive-width stripes.
Refer to Figure~\ref{fig:2d_I}.

\begin{figure}
  \centering
  \begin{subfigure}{\linewidth}
    \centering
    \includegraphics[scale=0.5]{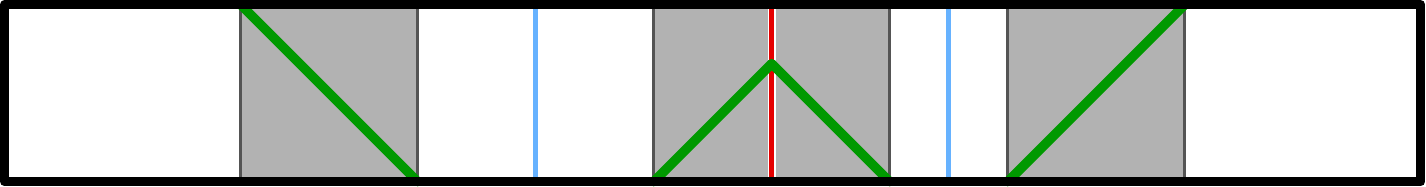}
  \end{subfigure}
  $$\downarrow$$
  \begin{subfigure}{\linewidth}
    \centering
    \includegraphics[scale=0.5]{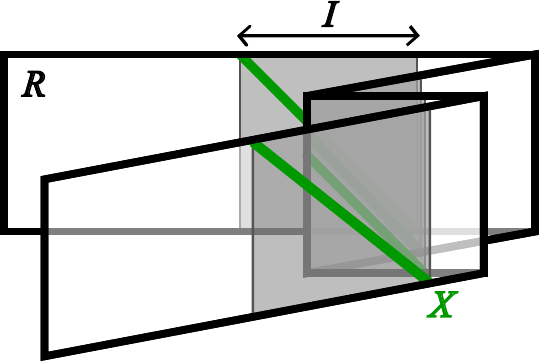}
  \end{subfigure}
    \caption{A small example instance and solution of orthogonal fold \& cut,
      illustrating the definition of the horizontal interval~$I$.}
  \label{fig:2d_I}
\end{figure}

\paragraph{Zero-width stripes.}
We argue that there must be a vertical crease at each vertical zero-width stripe.

Suppose for contradiction that a solution to orthogonal fold \& cut does
not put a crease at a vertical zero-width stripe.
Then the bands on opposite sides of the stripe remain on opposite sides
in the folded state, so there are cut segments on both sides of the stripe
in the folded state.
In particular, the intersection of the (diagonal) cut line with the stripe is
in the interior of the folded state, so some portion of the stripe
must lie directly above this intersection point,
and some portion must lie directly below it.

Because the stripe is connected, some layer in the folded state must pass the intersection point without a horizontal crease there. Because there is also not a vertical crease at the stripe, there is a neighborhood of the intersection point in that layer which has no creases. In particular, this neighborhood contains a cut segment that crosses the stripe not at a cut vertex, so our stripe must have been a band after all.

\paragraph{Positive-width stripes.}
Because the paper is an axis-aligned rectangle and the crease pattern is
orthogonal, the folded state is a rectangle~$R$,
with some diagonal cut line across it.
If we vertically project the portion $X$ of the cut line
that is within the folded rectangle~$R$,
then we obtain a horizontal interval~$I$.
Moreover, the image in the folded state of any vertical line $\ell$ crossing the unfolded paper
is a vertical segment crossing the entire folded state.
The image of $\ell$ in the folded state must intersect $X$
if and only if $\ell$ intersects a cut segment (including cut vertices).
In particular, vertical lines in bands (which cross cut segments)
must land in $I$, and vertical lines in the interior of positive-width
stripes (which do not intersect cut segments) must not land in~$I$.

Hence the vertical creases must result in all bands lying on a
horizontal open interval $I$, with no positive-width stripes
overlapping the interval.
When we also include the requirement that every zero-width stripe must have a crease,
this is exactly 1D interval fold \& cut:
bands are cut intervals, and stripes are empty intervals.
In the folded state, each boundary between a positive-width stripe and a band
must lie on one of the ends of~$I$.
Each positive-width stripe will be entirely on either the right or the left of
the interval, with its two boundary edges aligned on the same interval end.

Within a positive-width stripe, there may be many possible crease patterns,
but (just as in Lemma~\ref{lem:1d interval}) all that matters is that the boundaries of the stripe are aligned,
and the interior of the stripe does not pass the end of the interval $I$
with the cut~$X$.
All crease patterns satisfying this constraint
will result in the same cut pattern,
because they disagree only in a region that does not intersect the cut~$X$.
One such crease pattern is to simply put one vertical crease at the
center of each stripe, which is the canonical crease pattern.
\end{proof}

Thus, for $\alpha=1$, it suffices to compute the canonical crease pattern,
and then check whether it works.
This is fairly easy to implement algorithmically.
First, find the locations of vertical (respectively, horizontal)
stripes by looking at the horizontal (respectively, vertical) intervals
not covered by cuts.
When we fold on a crease in the center of a stripe,
the near edges of the two adjacent bands are aligned,
and the bands themselves are aligned until one of them ends.
So we just need to check that the cut pattern in the narrower adjacent band
exactly matches the cut pattern in the portion of the wider band
that is the reflection across the crease of the narrower band.
If this is the case for every vertical and horizontal stripe,
then the canonical crease pattern is a solution,
and otherwise the orthogonal fold \& cut instance is not solvable.

Combining Lemmas~\ref{lem:slopes} and \ref{lem:horizontal cuts} with Theorem~\ref{thm:of&c canonical works}, we have a complete characterization of solvable orthogonal fold \& cut instances on axis-aligned rectangular paper.

\section{Open Problems}
\label{sec:open}

There are several problems related to orthogonal fold \& cut that remain unsolved.

\begin{enumerate}
\item What if we allow nonrectangular paper,
  or non-axis-aligned rectangular paper,
  as we used in our font in Figure~\ref{fig:font}?
  Lemma~\ref{lem:orthogonal crease patterns} holds for any orthogonally convex shape or when restricting to simple folds, but our analysis of stripes and bands breaks down for anything other than an axis-aligned rectangle. Perhaps one could find a way to extend the cut pattern to the bounding axis-aligned rectangle. For paper that is not orthogonally convex with nonsimple folds allowed, the problem seems harder still: different ``arms'' of the paper could be folded in different ways.
\item What if we allow a slightly larger set of possible creases? For instance, one could require creases to be at multiples of $45^\circ$ (as in box pleating), or some other angle. An analog of Lemma~\ref{lem:slopes} still holds, but we lose the large amount of structure present in orthogonal foldings.
\item Our font in Figure~\ref{fig:font} leaves room for improvement in the clarity of some characters, and we welcome suggestions. It would also be interesting to develop two different typefaces, one where the folds and paper boundary are axis-aligned and one where the folds are all orthogonal to each other but not the paper boundary.
\end{enumerate}

\section*{Acknowledgments}

This work was initiated during an MIT class on Geometric Folding Algorithms
(6.849, Fall 2020).  We thank the other participants of that class---in particular, Jamie Tucker-Foltz, Naveen Venkat, Hanyu Zhang, Midori Zhou---for helpful discussions and providing a productive research environment.

\bibliographystyle{alpha}

\bibliography{paper}

\newcommand{\etalchar}[1]{$^{#1}$}
\begin{thebibliography}{ADD{\etalchar{+}}17}

\bibitem[ADD{\etalchar{+}}17]{asao2017folding}
Yasuhiko Asao, Erik~D. Demaine, Martin~L. Demaine, Hideaki Hosaka, Akitoshi Kawamura, Tomohiro Tachi, and Kazune Takahashi.
\newblock Folding and punching paper.
\newblock {\em Journal of Information Processing}, 25:590--600, 2017.

\bibitem[BDEH01]{bern2002disk}
Marshall Bern, Erik Demaine, David Eppstein, and Barry Hayes.
\newblock A disk-packing algorithm for an origami magic trick.
\newblock In {\em Origami$^3$: Proceedings of the 3rd International Meeting of Origami Science, Math, and Education}, pages 17--28. A K Peters, 2001.

\bibitem[DDL98]{demaine1998folding}
Erik~D. Demaine, Martin~L. Demaine, and Anna Lubiw.
\newblock Folding and cutting paper.
\newblock In {\em Revised Papers from the Japan Conference on Discrete and Computational Geometry ({JCDCG}'98)}, volume 1763, pages 104--117, 1998.

\bibitem[DO07]{Demaine-O'Rourke-2007}
Erik~D. Demaine and Joseph O'Rourke.
\newblock {\em Geometric Folding Algorithms: Linkages, Origami, Polyhedra}.
\newblock Cambridge University Press, July 2007.

\bibitem[K{\"{o}}n36]{Koenig-1936}
D\'enes K{\"{o}}nig.
\newblock {\em {Theorie der endlichen und unendlichen Graphen}}.
\newblock Akademische Verlagsgesellschaft (Leipzig), 1936.
\newblock See also the English translation, \emph{Theory of Finite and Infinite Graphs}, Birkhauser (Boston), 1990.

\end{thebibliography}

\end{document}